\documentclass[12pt,draftclsnofoot,onecolumn]{IEEEtran}
\usepackage{amsfonts}
\usepackage[pdftex, bookmarksnumbered, bookmarksopen, colorlinks, citecolor=black, linkcolor=black]{hyperref}
\usepackage{graphicx}
\usepackage{color}
\usepackage{amsmath,amsfonts,amssymb,amsthm,epsfig,epstopdf,url,array}
\usepackage{url,textcomp}
\usepackage{mdwmath}
\usepackage{mdwtab}
\usepackage{times}
\usepackage{authblk}
\usepackage{cite}
\usepackage{amsbsy}
\usepackage{subfigure}
\usepackage{bm}
\usepackage{doi}
\usepackage{soul}
\usepackage{cancel}
\usepackage{algorithm} 
\usepackage{algorithmic} 
\usepackage{multirow} 
\usepackage{amsmath}
\usepackage{xcolor}
\usepackage{tikz}
\newcommand*{\circled}[1]{\lower.7ex\hbox{\tikz\draw (0pt, 0pt)%
    circle (.42em) node {\makebox[0.1em][c]{\small #1}};}}

\newtheorem{theorem}{Theorem}
\newtheorem{lemma}{Lemma}
\newtheorem{Proposition}{Proposition}

\newtheorem{mydef}{Definition}

\begin{document}
\title{ Sum-Rate Maximization of RIS-Aided Multi-User MIMO Systems With Statistical CSI}
\author{Huan Zhang, Shaodan~Ma, Zheng Shi, Xin Zhao,
and Guanghua~Yang
\thanks{Huan Zhang and Shaodan Ma are with the State Key Laboratory of Internet
of Things for Smart City and the Department of Electrical and Computer
Engineering, University of Macau, Macao S.A.R. 999078, China (e-mails: cquptzh@gmail.com, shaodanma@um.edu.mo).}
\thanks{Zheng Shi and Guanghua Yang are with the Institute of Physical Internet and the School of Intelligent Systems Science
and Engineering, Jinan University, Zhuhai 519070, China (e-mails: shizheng0124@gmail.com, ghyang@jnu.edu.cn).}
\thanks{Xin Zhao is with the School of Information
and Electronics, Beijing Institute of Technology, Beijing 100081, China, and also with the Department of Electrical and Computer Engineering,
University of Macau, Taipa, Macao S.A.R. 999078, China (e-mail: xinzhao.eecs@gmail.com).}}

\maketitle
\begin{abstract}
This paper investigates a reconfigurable intelligent surface (RIS)-aided multi-user multiple-input multiple-output (MIMO) system by considering only the statistical channel state information (CSI) at the base station (BS). We aim to maximize its sum-rate via the joint optimization of beamforming at the BS and phase shifts at the RIS. However, the multi-user MIMO transmissions and the spatial correlations make the optimization cumbersome. For tractability, a deterministic approximation is derived for the sum-rate under a large number of the reflecting elements. By adopting the approximate sum-rate for maximization, the optimal designs of the transmit beamforming and the phase shifts can be decoupled and solved in closed-forms individually. More specifically, the global optimality of the transmit beamforming can be guaranteed by using water-filling algorithm and a sub-optimal solution of phase shifts can be obtained by using projected gradient ascent (PGA) algorithm. By comparing to the case of the instantaneous CSI assumed at the BS, the proposed algorithm based on statistical CSI can achieve comparable performance but with much lower complexity and signaling overhead, which is more affordable and appealing for practical applications. Moreover, the impact of spatial correlation is thoroughly examined by using majorization theory. 
\end{abstract}
\IEEEpeerreviewmaketitle
\section{Introduction}
\subsection{Background}
\IEEEPARstart{M}assive multiple-input multiple-output (MIMO), millimeter wave (mmWave) communications, ultra-dense networks (UDN), and non-orthogonal multiple access (NOMA) have been recognized as the key enabling techniques to achieve higher spectral efficiency and massive machine-type communications (mMTCs) for the next generation mobile cellular systems \cite{larsson2014massive,pi2011introduction,kamel2016ultra,shi2019zero}. Particularly, by deploying a massive number of antennas over 30-300 GHz bands, the mmWave massive MIMO communication paradigm has shown its great potentials for providing ultra-high data-rate transmission as well as massive device connectivity in the upcoming 5G systems \cite{busari2017millimeter}.
Although the mmWave massive MIMO has attracted numerous interests from both the academic and industry communities, the confronted challenges span the broad fields of communication engineering and allied disciplines, such as escalating signal processing complexity, short-range communications, increasing hardware costs, high power consumption, etc. Moreover, the UDN and NOMA techniques also bring the additional deployment costs and imperfect successive interference cancellation (SIC) bottleneck, respectively. To meet the need for economic and sustainable future cellular networks, more efficient techniques have been proposed to address the spectrum shortage problem and improve the system performance.

With the development of radio frequency (RF) micro electromechanical systems (MEMS) and metamaterial (e.g., metasurface), reconfigurable intelligent surfaces (RIS) is a brand-new technology which has gained significant momentum \cite{wu2020towards}. Specifically, RIS is an artificial metamaterial, which is composed of a large number of low-cost, passive and reflective elements with reconfigurable parameters. It is capable of digitally manipulating electromagnetic waves and intelligently adjusting the reflecting signal phase to achieve certain communication objectives, e.g., enhancing
the received signal energy, expanding the coverage region, alleviating interference, etc. \cite{ElMossallamy2020}. By comparing with conventional cooperative communication systems, i.e., amplify-and-forward (AF), decode-and-forward (DF), and compress-and-forward (CF) relaying schemes, RIS assisted wireless communications can obtain preferable electromagnetic propagation environment with limited power consumption.
\subsection{Related Works}
 By properly adjusting the phase shifts to improve the electromagnetic propagation environment with a low cost and energy consumption, the RIS has attracted tremendous attention from both industry and academics, recently \cite{wu2020towards,ElMossallamy2020,wu2019intelligent,wu2020beamforming,guo2020weighted,zhou2020intelligent,ye2020joint,pan2020multicell,zhang2020sum,pan2020intelligent, zheng2020intelligent,
mu2020exploiting,zhao2020performance,wang2020intelligent,li2020reconfigurable,zhu2020Stochastic,mishra2019channel,taha2019enabling,han2019large,zhang2020transmitter,wang2021joint,
chaaban2020opportunistic,alwazani2020intelligent,kammoun2020asymptotic}.
The authors in \cite{wu2019intelligent} proposed RIS-aided multiple-input single-output (MISO) system communications by jointly optimizing the transmit beamforming at the base station (BS) and reflect beamforming at the RIS to minimize the total transmit power consumption. Numerical results showed that the RIS not only achieves the comparable performance to the conventional massive MIMO or multi-antenna AF relay, but also is more cost efficient. Similar conclusions were drawn in \cite{wu2020beamforming} by extending the continuous phase shifts to discrete phase shifts due to the hardware limitation at reflecting elements. In \cite{guo2020weighted} and \cite{zhou2020intelligent}, the authors considered the sum-rate maximization problem for RIS-assisted multi-group multi-cast MISO and  multi-user MISO systems to find the optimal solution for the precoding matrix at the BS and the reflection coefficients at the RIS under both the power and unit-modulus constraints. In order to reap the benefit of diversity gain at user sides, the symbol error rate (SER) minimization problem for the RIS-assisted MIMO systems and the sum-rate maximization problem for the RIS-assisted multi-user MIMO systems also have been investigated in \cite{ye2020joint} and \cite{pan2020multicell}, respectively.
Furthermore, the RIS-assisted user cooperation has been extended to the wireless powered communications, NOMA network, backscatter communications, physical layer security, unmanned aerial vehicle (UAV) communications, millimeter wave networks, etc. \cite{pan2020intelligent,zheng2020intelligent,mu2020exploiting,zhao2020performance,wang2020intelligent,li2020reconfigurable,zhu2020Stochastic}.

Although the wireless communication systems can be significantly improved with the help of the RIS, the aforementioned works for RIS-aided  communication systems design, i.e., \cite{wu2019intelligent,wu2020beamforming,guo2020weighted,zhou2020intelligent,ye2020joint,pan2020multicell,pan2020intelligent,zhang2020sum,zheng2020intelligent,mu2020exploiting,zhao2020performance,wang2020intelligent,li2020reconfigurable,zhu2020Stochastic}, are all based on instantaneous channel state information (CSI) available at BS. As opposed to the traditional wireless systems where channel acquisition is a straightforward matter, due to the passive nature of RIS, the low-cost reflecting elements can not possess any active RF chains to facilitate channel estimation. In \cite{mishra2019channel}, the authors estimated the reflecting cascaded channel via element-wise on-off operation at each reflecting element. Another alternative approach to estimate the channel coefficients is leveraging compressive sensing and deep learning tools by proposing a novel RIS architecture, where sparse channel sensors are assumed  and a few number of elements are connected to the baseband of the RIS controller \cite{taha2019enabling}. However, the channel training overhead becomes excessively high as the number of RIS reflecting elements and/or RIS-aided users increases. So  the authors studied the RIS-aided communication systems design basing on statistical CSI in \cite{han2019large,zhang2020transmitter,wang2021joint}, by considering the practical CSI acquisition.  Furthermore, another difficulty for RIS-aided wireless communication systems is the joint design of the transmit precoding and the phase shift (active and passive beamforming) design at BS and RIS. Since the optimization problems are usually non-convex, the alternating algorithms have been widely adopted for sub-optimal solutions in \cite{wu2019intelligent,wu2020beamforming,guo2020weighted,zhou2020intelligent,ye2020joint,pan2020multicell,zhang2020sum,pan2020intelligent, zheng2020intelligent,mu2020exploiting,zhao2020performance,wang2020intelligent,li2020reconfigurable,zhu2020Stochastic,han2019large,zhang2020transmitter,wang2021joint} with high computational complexity.

Meanwhile, the spatial correlation between antenna or reflecting elements for RIS-aided system always exist in realistic propagation environments due to mutual antenna coupling, space limitations for adjacent antenna or reflecting elements, and non-rich scattering environments, etc. \cite{wang2021joint,chaaban2020opportunistic,kammoun2020asymptotic,alwazani2020intelligent,yang2020asymptotic}. However, for analytical simplicity, the RIS-aided MISO/MIMO systems for multi-user case in existing works, i.e., \cite{wu2020towards,ElMossallamy2020,wu2019intelligent,wu2020beamforming,zhou2020intelligent,guo2020weighted,pan2020multicell,zhang2020sum,ye2020joint, pan2020intelligent,zheng2020intelligent,mu2020exploiting,zhao2020performance,wang2020intelligent,li2020reconfigurable,zhu2020Stochastic,mishra2019channel,taha2019enabling,han2019large,zhang2020transmitter}, assumed the independence among antenna and reflecting elements. Even though the spatial correlation has been taken into account for RIS in \cite{chaaban2020opportunistic,kammoun2020asymptotic,alwazani2020intelligent}, those works only considered single-side spatial correlation on RIS (from RIS to users link) for single-input single-output (SISO) or MISO communication systems.
\subsection{Motivation and Contributions}
Motivated by the above challenges and issues, we focus on the maximization of the sum-rate of an RIS-assisted multi-user MIMO systems by leveraging only statistical CSI, in which the spatially correlated channels are considered. The main contributions of this paper can be summarized as follows.
\begin{itemize}
  \item  To reap the benefit of diversity gain from MIMO and RIS, we consider an RIS-aided multi-user MIMO system. Unlike the aforementioned works, the spatial correlations at BS, RIS and users are considered in the paper. Based on this practical system model, the random matrix theory is invoked to derive the deterministic equivalent approximation for the sum-rate.
  \item Thanks to the deterministic equivalent results, the impacts of the spatial correlation on each terminal are thoroughly investigated with the help of majorization theory. It is obtained that the approximation sum-rate is irrelevant to the antenna correlation at the BS, while it significantly depends on the spatial correlation at user and the transceiver sides at RIS as the number of reflecting elements increases. More specifically, the spatial correlation at user has a negative impact on sum-rate, but the sum-rate benefits from the transceiver sides' correlation at RIS.
  \item In contrast to \cite{wu2020towards,ElMossallamy2020,wu2019intelligent,wu2020beamforming,
      zhou2020intelligent,guo2020weighted,pan2020multicell,zhang2020sum,ye2020joint,
      pan2020intelligent,zheng2020intelligent,mu2020exploiting,zhao2020performance,
      wang2020intelligent,li2020reconfigurable,zhu2020Stochastic,mishra2019channel,
      taha2019enabling,han2019large,zhang2020transmitter} that assumed instantaneous CSI, the statistical CSI is exploited to maximize the sum-rate in this paper, including the path-loss coefficients, the second-order statistics of small-scale fading, i.e., spatial correlation matrices. 
      By comparing with the assumption of the instantaneous CSI, the statistical CSI changes slowly and does not require high signaling overhead.
  \item As opposed to \cite{wang2021joint} where the joint design of the transmit beamforming and the phase shifts are highly coupled and a sub-optimal solution was obtained by using the alternative algorithm, we resort to the deterministic approximation to decouple the joint optimization problem into two sub-problems, i.e., the transmit beamforming design and the phase shifts optimization. The two sub-problems can be solved in closed-forms with two efficient algorithms, respectively, i.e., water filling algorithm and projected gradient ascent (PGA) algorithm. Since the proposed algorithms don't require the knowledge of instantaneous CSI, it is more practical especially for the case of time-varying channels, where the instantaneous channel state estimation is hardly impossible for passive reflecting elements. Moreover, the global optimality of the proposed transmit beamforming design is guaranteed.
\end{itemize}
\subsection{Structure of the Paper}
The rest of the paper is organized as follows. In Section II, the system model of RIS aided multi-user MIMO communication under spatial correlation is presented. Section III derives the deterministic equivalent results for sum-rate, with which the impact of spatial correlation is thoroughly studied. In section IV, the asymptotic sum-rate is maximized by jointly optimizing the transmit beamforming matrix and phase shift vector. In Section V, we present and discuss the numerical results and the Monte-Carlo simulations. Finally, Section VI concludes the paper.

\emph{Notations}: The following notations are used throughout this paper. We use lower-case and upper-case letters to indicate column vectors and matrices, respectively. $\mathbf{I}_N$ denotes an $N\times N$ identity matrix. The superscripts $(\cdot)^{\rm{T}}$ and $(\cdot)^{\mathrm{H}}$ denote the transpose and the transpose-conjugate operations, respectively. The notations $(\cdot)^{\frac{1}{2}}$, $\mathrm{Tr}(\cdot)$ and $\mathrm{det}(\cdot)$ represent the matrix principal square root, trace and determinant of matrices, respectively. $(\cdot)^\ast$ denotes complex conjugate.
$\mathbb{C}$ stands for the set of complex numbers. $\mathbf{A}(i,j)$ refers to the $(i,j)$-th element of matrix $\mathbf{A}$. We denote $\mathbf{A}\succeq0$ if $\mathbf{A}$ is positive semidefinite. $\mathbb{E}(\cdot)$ is the statistical expectation operator.
%
%
%

\section{System Model and Problem Formulation}
\subsection{Transmission Model}
\begin{figure}[!h]
\centering
\includegraphics[width=3.5in]{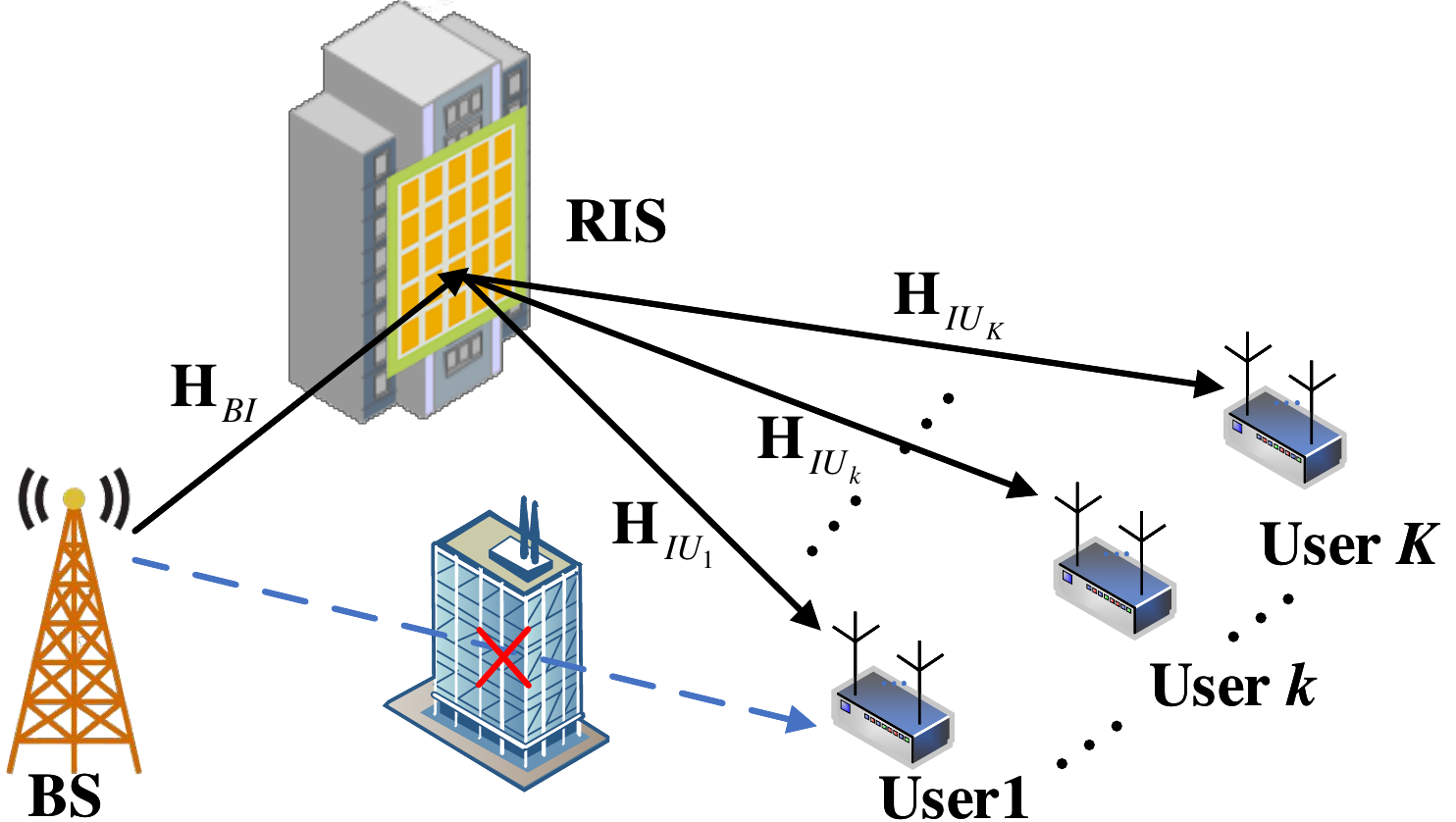} 
\caption{An RIS-aided multi-user MIMO downlink communication system.}
\label{system_model}
\end{figure}
This paper investigates an RIS-aided multi-user MIMO communication systems as shown
in Fig. \ref{system_model}, where a multi-antenna BS serves $K$ users with the help of a multi-reflecting element RIS. We assume that the BS and each user are equipped with $M$ and $L$ antennas, respectively,  and the users are located in the vicinity of the RIS. 
To enhance the system performance, a uniform planar
array (UPA) RIS composed of a large number of passive and low-cost reflecting elements is mounted on the wall of a surrounding high-rise building to assist the BS in communicating with the users, and the massive number of reflecting elements at the RIS is $N$, i.e., $N\gg M\geq L$. Similarly to \cite{wang2021joint} and \cite{kammoun2020asymptotic}, we assume that the direct links from the BS to users are unavailable due to the obstacles, such as buildings. 
The channel matrices from the BS to the RIS and from the RIS to the $k$-th user are denoted by $\mathbf{H}_{BI}\in \mathbb{C}^{N\times M}$ and $\mathbf{H}_{IU_{k}}\in\mathbb{C}^{L\times N}$, respectively, where $k\in \{1,\cdots,K\}$. We denote the diagonal reflection matrix of RIS as $\mathbf{\Theta}=\mathrm{diag}(e^{j\theta_1},\cdots,e^{j\theta_N})$, the phase shift vector by $\boldsymbol{\theta}=[\theta_1,\cdots,\theta_N]^{\mathrm{T}}$, where $\theta_n\in [0,2\pi]$ represents phase shift after reflecting by the $n$th element. Accordingly, the received signal at user $k$ is given by
\begin{equation}\label{y_k}
\mathbf{y}_k = \mathbf{H}_{IU_{k}}\mathbf{\Theta} \mathbf{H}_{BI}\mathbf{W}\mathbf{s}+\mathbf{n}_k,
\end{equation}
where
 \begin{itemize}
   \item $\mathbf{s}=[\mathbf{s}_1^{\mathrm{T}},\cdots,\mathbf{s}_k^{\mathrm{T}},\cdots,\mathbf{s}_K^{T}]^\mathrm{T}\in \mathbb{C}^{Kd\times1}$, $\mathbf{s}_k \in  \mathbb{C}^{d\times 1}$ denotes the transmit symbol vector of user $k$ with $d$ data steams and satisfying $1\leq Kd \leq {\rm min}(M,KL)$. Further, it is assumed that  $\mathbb{E}(\mathbf{s}_i\mathbf{s}_i^\mathrm{H})=\mathbf{I}_{d}$, and $\mathbb{E}(\mathbf{s}_i\mathbf{s}_j^\mathrm{H})=\mathbf{0}$, for $i\neq j$.
   \item
        $\mathbf{W}=[\mathbf{W}_1,\cdots,\mathbf{W}_k,\cdots,\mathbf{W}_K]$ is the $M\times Kd$ transmit  beamforming matrix at the BS, in which $\mathbf{W}_k\in \mathcal{C}^{M\times d}$ is  the linear precoding matrix for the $k$th user and with total transmit power constraint $\mathrm{Tr}(\mathbf{W}\mathbf{W}^\mathrm{H})=\sum\limits_{k = 1}^K{\mathrm{Tr}}(\mathbf{W}_k\mathbf{W}_k^\mathrm{H})\leq P_{max}$.
   \item $\mathbf{n}_k\sim {\cal{CN}}(\mathbf{0}, \sigma_k^2\mathbf{I}_{L})$ represents the complex Gaussian noises at user $k$.
 \end{itemize}

Then, the instantaneous achievable rate for the $k$ user is given by
\begin{equation}\label{R_k}
{R}_k= {\log _2}\det\left[ {{\bf{I}}_L} + \mathbf{\Psi}_k{\left(\sum\limits_{i = 1,i\neq k}^K \mathbf{\Psi}_i+
{{\sigma_k^2}\mathbf{I}_L}\right)}^{ - 1} \right],
\end{equation}
where $\mathbf{\Psi}_j=\mathbf{H}_{IU_{k}}\mathbf{\Theta} \mathbf{H}_{BI}\mathbf{W}_j(\mathbf{H}_{IU_{k}}\mathbf{\Theta} \mathbf{H}_{BI}\mathbf{W}_j)^\mathrm{H}, j\in\{k,i\}$.
\subsection{Spatially Correlated Channels}
Due to multiple antennas and multiple reflecting elements mounted at users, BS and RIS, the close spacing among antennas/reflecting elements will result in the spatial correlation in realistic propagation environments. The spatial correlation significantly affects the achievable rate, which cannot be ignored especially for limited size of RIS and terminals \cite{kammoun2020asymptotic,yang2020asymptotic,zhang2018space}. To account for the impact of spatial correlation, the channel matrices $\mathbf{H}_{BI}$ and $\mathbf{H}_{IU_{k}}$ are modeled by the widely used Kronecker correlation channel model as\cite{yang2020asymptotic}\footnote{
Although we assume that all users have the common channel correlation matrices at RIS for notational simplicity, i.e., $\mathbf{T}_I$ and $\mathbf{R}_I$, the analytical result can be easily extended to the case of the users with different correlation matrices by replacing $\mathbf{T}_{I_k}$ and $\mathbf{R}_{I_k}$ with $\mathbf{T}_I$ and $\mathbf{R}_I$, respectively.}
\begin{equation}\label{channel_model}
\mathbf{C}=\beta_{\bf C}^{\frac{1}{2}}\mathbf{R}_{\bf C}^{\frac{1}{2}}\widetilde{\mathbf{C}}\mathbf{T}_{\bf C}^{\frac{1}{2}}, \,\mathrm{for}\,(\mathbf{C},\widetilde{\mathbf{C}})\in\{(\mathbf{H}_{BI},\widetilde{\mathbf{H}}_{BI}),
(\mathbf{H}_{IU_{k}},\widetilde{\mathbf{H}}_{IU_{k}})\},
\end{equation}
where the channel model in \eqref{channel_model} is composed of four independent components, including $\beta_{\bf C}$, $\widetilde{\mathbf{C}}$, $\mathbf{R}_{\bf C}$ and $\mathbf{T}_{\bf C}$. Specifically, $\beta_{\bf C}$ accounts for the (long-term) path loss effect. $\widetilde{\mathbf{C}}$ is a random matrix with independent identically distributed (i.i.d.) Gaussian random entries of zero mean and unity variance, i.e., ${\rm{vec}}(\widetilde{\mathbf{C}})\sim{\cal{CN}}({\mathbf{0}},{\mathbf{I}}\otimes{\mathbf{I}})$,
$\widetilde{\mathbf{H}}_{BI}\in\mathbb{C}^{N\times M}$, and $\widetilde{\mathbf{H}}_{IU_k}\in\mathbb{C}^{L\times N}$. $\mathbf{T}_{\bf C}$ and $\mathbf{R}_{\bf C}$ are respectively termed as the transmit/reflect and receive correlation matrices, $\mathbf{T}_{\bf C}\in\{\mathbf{T}_B,\mathbf{T}_I\}$ and $\mathbf{R}_{\bf C}\in\{\mathbf{R}_I,\mathbf{R}_{U_k}\}$, where $\mathbf{T}_B\in\mathbb{C}^{M\times M}$ and  $\mathbf{T}_I\in\mathbb{C}^{N\times N}$ stand for the transmit correlation matrix at the BS and transmit side correlation matrix at RIS, respectively, and $\mathbf{R}_I\in\mathbb{C}^{N\times N}$ and $\mathbf{R}_{U_k}\in\mathbb{C}^{L\times L}$ represent the correlation matrix at the receive side of RIS and the $k$th user, respectively. In addition, all the correlation matrices are positive semi-definite Hermitian matrices \cite{noh2014pilot}. For analytical simplicity, we assume those correlation matrices are normalized such that $\mathrm{Tr}(\mathbf{T}_{B})=M$, $\mathrm{Tr}(\mathbf{T}_{I})=\mathrm{Tr}(\mathbf{R}_{I})=N$ and $\mathrm{Tr}(\mathbf{R}_{U_k})=L$. Moreover, it is worth noting that all the mentioned spatial correlated matrices can be assumed to be time-invariant, because the antenna array response in dense scattering environments changes slowly \cite{noh2014pilot}.
\subsection{Problem Formulation}
In this paper, we focus on the optimal design of transmit beamforming matrix $\mathbf{W}^\star$ at BS and phase shift vector $\boldsymbol{\theta}^\star$ at RIS by maximizing the sum-rate as
\begin{equation}
\begin{aligned}
 \label{P}
\mathcal{P}\,:\,&\max\limits_{\mathbf{W},\boldsymbol{\theta}}\quad \sum\limits_{k = 1}^K R_k, \\
& \begin{array}{r@{\quad}r@{}l@{\quad}l}
\mathrm{s.t.}
&(\mathrm{C1}):&\mathrm{Tr}(\mathbf{W}\mathbf{W}^\mathrm{H})\leq P_{max},\\
&(\mathrm{C2}):&0\leq\theta_n\leq 2\pi, n=1,\cdots,N,
\end{array}
\end{aligned}
\end{equation}
where (C1) and (C2) represent the transmit power constraint and phase shift range constraint, respectively.

It is obviously found that problem $\mathcal{P}$ is difficult to solve, due to the coupling relationship between the precoding matrix $\bf W$ and phase shift vector $\boldsymbol{\theta}$, and non-convexity of the phase shift constraint.  To tackle this
challenging problem, most of works  resort to an alternating way to optimize the
transmit beamforming and phase shift until the convergence is reached, e.g.  \cite{wu2019intelligent,wu2020beamforming,zhou2020intelligent,guo2020weighted,pan2020multicell,zhang2020sum,ye2020joint,pan2020intelligent}. Although the optimization variables can be decoupled with the help of alternating optimization and the subproblem at the $t$-th iteration can be efficiently solved, the objective functions in  \cite{wu2019intelligent,wu2020beamforming,zhou2020intelligent,guo2020weighted,pan2020multicell,zhang2020sum,ye2020joint,pan2020intelligent}  are derived under the requirement of instantaneous CSI at both the BS and RIS. However, this assumption is impractical for a large scale number of antennas and reflecting elements. In practice, since the CSI is usually estimated with the aid of pilot sequences, the passive RIS hardware constraint entails prohibitively high channel estimation accuracy and frequent signal exchange overhead between BS and RIS\cite{mishra2019channel,taha2019enabling}.
Hence, the prior alternating optimization algorithms relying on instantaneous CSI have to face the critical issue of channel estimation and high phase adjustment costs for RIS-assisted systems. Besides, it is more intractable to maximize the instantaneous achievable rate in (\ref{P}), by comparing to RIS-assisted systems with single antenna receiver or single user in \cite{wu2019intelligent,wu2020beamforming,zhou2020intelligent,guo2020weighted,zhang2020sum,ye2020joint,pan2020intelligent,chaaban2020opportunistic,kammoun2020asymptotic,alwazani2020intelligent}, due to the involvement of multi-data streams, multi-antenna transceivers and multiple users.

To overcome the above issues, we assume that the RIS is equipped with a large number of reflecting elements ($N\rightarrow \infty$) to exploit the channel hardening property for RIS.  Thanks to the low energy cost features of RIS-assisted system, hundreds of reflecting elements for RIS are more affordable by comparing to hundreds of antennas equipped in massive MIMO systems in practice \cite{tang2020wireless}. In such circumstance, we derive a deterministic equivalent expression for the sum-rate by using random matrix theory tools. The analytical result enables us to reformulate a new optimization problem by exploiting statistical CSI as well as gain some meaningful insights.
\section{Analysis of Achievable Rate}

\subsection{Deterministic Equivalent Analysis} \label{single_cell}
Since the multi-user MIMO, multi-data streams and RIS cascaded channel are considered, the instantaneous sum-rate expression involves complicated matrix operations such as matrix inverse, determinant, and matrix chain products, which significant challenges on the jointly optimization design and mathematical analysis.
In order to obtain a tractable optimization problem and gain more insights, prior to designing the transmit beamforming matrix and the phase shifts vector, we first derive the deterministic sum-rate for the RIS-assisted multi-user MIMO systems. The deterministic  approximation sum-rate states in Theorem \ref{theorem1}.
\begin{theorem}\label{theorem1}
As the number of reflecting elements at the RIS increases, i.e. $N \to \infty $, a deterministic approximation of the sum-rate is given by
\begin{equation}\label{R_k_asy}
\begin{aligned}
R&\xrightarrow[{N \to \infty }]{{a.s.}}\sum\limits_{k = 1}^K{\log _2}\det \bigg[ {{\bf{I}}_L}+ \mathbf{\Phi}_k{ \big(\sum\limits_{i = 1,i\neq k}^K\mathbf{\Phi}_i+
{\frac{{\sigma_k^2}}{\varphi_k{\mathrm{Tr}\left(\mathbf{T}_I\mathbf{\Theta}\mathbf{R}_I\mathbf{\Theta}^{\mathrm{H}}\right)}}\mathbf{I}_L} \big)}^{ - 1} \bigg]\\
&\triangleq R^{asy},
\end{aligned}
\end{equation}
where $\mathbf{\Phi}_j=\frac{1}{N}\mathrm{Tr}\left(\mathbf{T}_B {\mathbf{W}_j}{\mathbf{W}_j^\mathrm{H}}\right){ \mathbf{R}}_{U_{k}}, j\in\{k,i\}$ and  $\varphi_k={N}\beta_{\mathbf{H}_{IU_{k}}}\beta_{\mathbf{H}_{BI}}$.
\begin{proof}
See Appendix \ref{app:proof_Theorem1}.
\end{proof}
\end{theorem}
By considering a large number of reflecting elements in RIS-assisted systems, it is worth noting that the instantaneous sum-rate converges to a deterministic equivalent which only depends on the statistical CSI, including path-loss coefficients and the spatial correlation matrices. 
Similar to the massive MIMO systems, the deterministic property for the large scale RIS-assisted wireless communication systems is due to the channel hardening effect, in which the channel variations become small and the fading channels turn to be deterministic with the increase of the number of reflecting elements on the RIS \cite{hochwald2004multiple}.

With the help of random matrix theory, the sum-rate for  RIS-aided multi-user MIMO scenarios can be represented by tractable deterministic equivalent as (\ref{R_k_asy}). Section \ref{section_V} will verify that the approximate result matches well with Monte Carlo simulations if the RIS equips with more than 200 reflecting elements.\footnote{It is worth noticing that the  deterministic equivalent conclusions for massive MIMO systems have been  investigated in \cite{feng2017power,feng2018impact,feng2019two}. Although the channel hardening property is also applicable to massive MIMO systems, the deployment costs for hundreds of antenna array are very high. Comparing with Massive MIMO, the large scale RIS with hundreds of passive reflecting element requirements are more affordable and its prototype has already been designed in \cite{tang2020wireless}.} However, a large number of reflecting elements on RIS will cause serious spatial correlation. Hence, it is necessary and interesting to investigate the impact of spatial correlation for RIS-aided MIMO systems, and some meaningful insights will be provided.
\subsection{Impact of Spatial Correlation}\label{correlation_impact}
From the deterministic equivalent of the sum-rate in (\ref{R_k_asy}), the asymptotic sum-rate is irrelevant to the transmit antenna correlation at BS $\mathbf{T}_B$, but depends heavily on the antenna correlation at the RIS and user sides through $\mathrm{Tr}(\cdot)$ and $\mathrm{det}(\cdot)$, respectively. However, the transmit beamforming matrix $\mathbf{W}$ and  diagonal reflection matrix $\mathbf{\Theta}$ in these matrix operations cause difficulties to investigate the impact of spatial correlation analytical. To thoroughly investigate the impact of spatial correlation, we assume equal power allocation (EPA) at the BS and identity reflection matrix, i.e., $\mathbf{\Theta}=\mathbf{I}_N$, to ease the analysis.  Accordingly, the asymptotic sum-rate can be simplified as
\begin{equation}\label{R_k_asy_corr}
\begin{aligned}
\bar{R}^{asy}&=\sum\limits_{k = 1}^K{\log _2}\det\bigg[ {{\bf{I}}_L} +\epsilon\mathbf{R}_{U_k} {\left(({K - 1})\epsilon\mathbf{R}_{U_k}+
{\frac{{\sigma_k^2}}{\varphi_k{\rm{Tr}}\left( {{\bf{T}}_I}{{\bf{R}}_I} \right)}\mathbf{I}_L}\right)}^{ - 1}\bigg],
\end{aligned}
\end{equation}
where $\epsilon=\frac{P_{max }}{KN}$.

  Furthermore, to provide qualitative conclusion on the impact of the spatial correlation, some basic concepts are introduced to facilitate our analysis, including majorization, Schur-convexity, and Schur-concavity.
\begin{mydef}\label{def1}
(\textbf{Majorization\cite{marshall1979inequalities}}) For two positive semi-definite matrices $\mathbf{A}_{1}$ and $\mathbf{A}_{2}$ with identical dimension $N\times N$,  the descending order vector $\boldsymbol{\lambda}_1=(\lambda_{1,1},\cdots,\lambda_{1,{N}})^\mathrm{T}$ and $\boldsymbol{\lambda}_2=(\lambda_{2,1},\cdots,\lambda_{2,{N}})^\mathrm{T}$ represent the vectors of the eigenvalues for $\mathbf{A}_{1}$ and $\mathbf{A}_{2}$, respectively. If
 \begin{equation}
\begin{aligned}
\left\{ \begin{array}{l}
\sum\limits_{i = 1}^k {{\lambda_{1,i}}}  \ge \sum\limits_{i = 1}^k {{\lambda_{2,i}}} ,{\mkern 1mu} {\mkern 1mu} {\mkern 1mu} k = 1,2, \ldots ,N - 1{\mkern 1mu} {\mkern 1mu} \\
\sum\limits_{i = 1}^N {{\lambda_{1,i}}}  = \sum\limits_{i = 1}^N {\lambda_{2,i}}
\end{array} \right.
\end{aligned},
\end{equation}
  we say $\boldsymbol{\lambda}_1\succeq\boldsymbol{\lambda}_2$ and the matrix $\mathbf{A}_{1}$ is more correlated than the matrix $\mathbf{A}_{2}$ \cite{huan2020rate}. Therefore, the eigenvalue vector $[\underbrace{N,0,\cdots,0}_N]^\mathrm{T}$ and $[\underbrace{{1},1,\cdots,1}_N]^\mathrm{T}$ correspond to full- and independent-correlated channel models w.r.t. matrix $\mathbf{A}_1$.
\end{mydef}
\begin{mydef}(\textbf{Schur-Convexity and Schur-Concavity\cite{marshall1979inequalities}})\label{def2}
For any two arbitrary vectors $\mathbf{x},\mathbf{y}\in \mathbb{R}^n$ and real valued function $f:\mathbb{R}^n\to \mathbb{R}$, if $\mathbf{x}\succeq \mathbf{y}$, $f$ is said to be Schur-Convexity for $f(\mathbf{x})\geq f(\mathbf{y})$. Conversely, $f$ is said to be Schur-Concavity for $f(\mathbf{x})\leq f(\mathbf{y})$.
\end{mydef}

 In what follows, the impacts of spatial correlations at the $k$-th user and the RIS on the  asymptotic sum-rate will be examined individually.
\subsubsection{Impact of $\mathbf{R}_{U_k}$}
By capitalizing on eigenvalue decomposition (EVD) of $\mathbf{R}_{U_k}=\mathbf{U}_{U_k}\mathbf{\Lambda}_{U_k} \mathbf{U}_{U_k}^\mathrm{H}$ into (\ref{R_k_asy_corr}), where $\mathbf{U}_{U_k}\in \mathbb{C}^{L\times L}$ is a unitary matrix and $\mathbf{\Lambda}_{U_k}={\rm diag}(\lambda_{U_{k_1}},\cdots,\lambda_{U_{k_L}})$ denotes a diagonal matrix and its diagonal elements are the eigenvalues of $\mathbf{R}_{U_k}$. Then, for deterministic matrices $\mathbf{T}_I$ and $\mathbf{R}_I$, the asymptotic sum-rate w.r.t. $\mathbf{R}_{U_k}$ is re-expressed as
\begin{equation}
\bar{R}^{asy}(\mathbf{R}_{U_k})=\sum\limits_{k = 1}^K\sum\limits_{l = 1}^L{\log _2}\big[{1 + ({K - 1 + \frac{\zeta _k}{\epsilon\lambda_{U_{k_l}}}}})^{-1}\big],
\end{equation}
where $\zeta_k=\frac{\sigma_k^2}{\varphi_k{\rm{Tr}}\left( {{\bf{T}}_I}{{\bf{R}}_I} \right)}$ and $\frac{\zeta_k}{\epsilon}>0$. Since the function $f({\lambda_{U_{k_l}}})={\log _2}\left[{1 + ({K - 1 + \frac{\zeta_k}{\epsilon}{\lambda_{U_{k_l}}^{ - 1}}}})^{-1}\right]$ is concave and twice differentiable w.r.t. $\lambda_{U_{k_l}}$, by using  \cite[Proposition 3.C.1]{marshall1979inequalities}, we have the asymptotic sum-rate is Schur-concave w.r.t. $\boldsymbol{\lambda}_{U_k}=[\lambda_{U_{k_1}},\cdots,\lambda_{U_{k_L}}]^\mathrm{T}$. With the help of Definitions \ref{def1} and \ref{def2}, we directly have the following theorem.
\begin{theorem}\label{theorem2}
If $\mathbf{R}_{U_{k_1}}$ is more correlated than $\mathbf{R}_{U_{k_2}}$, i.e., $\boldsymbol{\lambda}_{U_{k_1}}\succeq\boldsymbol{\lambda}_{U_{k_2}}$. We have
\begin{equation}
\bar{R}^{asy}(\mathbf{R}_{U_{k_1}})\leq \bar{R}^{asy}(\mathbf{R}_{U_{k_2}}).
\end{equation}
\end{theorem}
With Theorem \ref{theorem2}, it can be found that the spatial correlation at the $k$-th user has a negative impact on the deterministic  equivalent of sum-rate.

\subsubsection{Impact of $\mathbf{T}_{I}$ and $\mathbf{R}_{I}$}
From  (\ref{R_k_asy_corr}), both $\mathbf{T}_{I}$ and $\mathbf{R}_{I}$ are in trace operation, i.e., $\mathrm{Tr}(\mathbf{T}_{I}\mathbf{R}_{I})$, and it is hard to investigate the impact of $\mathbf{T}_{I}$ or $\mathbf{R}_{I}$ even when the corresponding matrix ($\mathbf{R}_{I}$ or $\mathbf{T}_{I}$) is deterministic, since the trace function $\mathrm{Tr}(\cdot)$ is both Schur-concave and Schur-convex w.r.t. the eigenvalue vector of $\mathbf{T}_{I}\mathbf{R}_{I}$. For analytical simplicity,  we assume $\mathbf{R}_I=\mathbf{T}_I$ and   $\mathrm{Tr}(\mathbf{T}_{I}\mathbf{R}_{I})=\mathrm{Tr}(\mathbf{T}_{I}^2)$. Evidently from Property 5 in \cite{feng2018impact}, we obtain that $\mathrm{Tr}(\mathbf{T}_{I}^2)$ is Schur-convex w.r.t. $\boldsymbol{\lambda}_{\mathbf{T}_I}$. By using the monotonically increasing property of $\bar{R}^{asy}({\mathbf{T}_I})$ w.r.t.  $\mathrm{Tr}(\mathbf{T}_{I}^2)$, the impact of  reflect-correlated matrix  $\mathbf{T}_{I}$ on the asymptotic sum-rate can be revealed in the following theorem.
      \begin{theorem}\label{theorem3}
If $\mathbf{T}_{I_1}$ is more correlated than $\mathbf{T}_{I_2}$, i.e., $\boldsymbol{\lambda}_{{T}_{I_1}}\succeq\boldsymbol{\lambda}_{{T}_{I_2}}$. We have
\begin{equation}
\bar{R}^{asy}(\mathbf{T}_{I_1})\geq \bar{R}^{asy}(\mathbf{T}_{I_2}).
\end{equation}
\end{theorem}
Hence, it can be found that the spatial correlation at the RIS has a positive impact on the deterministic equivalent of the sum-rate for $\mathbf{R}_I=\mathbf{T}_I$. This is due to the fact that the RIS can be regarded as a mirror for identical reflection matrix, and higher spatial correlation empowers more power focusing at receiver sides.

With the help of the deterministic equivalent conclusion in (\ref{R_k_asy}), the impacts of the spatial correlations at different sides can be quantified. The spatial correlations at the users and RIS have opposite impacts. More specifically, on the one hand, increasing spatial correlation on the user side reduces the effective dimensionality of the receiver. On the other hand, increasing spatial correlation on the RIS side also enables focusing power. Those results for spatial correlation impacts are consistent with  \cite{tulino2005impact}.
In the next section, we reformulate the optimization problem by using the deterministic approximation of the sum-rate.
\section{Sum-Rate Maximization}\label{new_problem}
Although one joint optimization for the RIS-aided multi-user MIMO system has been proposed in \cite{pan2020multicell}, the high signaling overhead of instantaneous CSI acquisition and computational complexity of alternating optimization algorithm hinder its application in practice. To overcome these issues, we resort to maximize the approximate sum-rate (\ref{R_k_asy}) instead. By comparing with previous works that maximized the sum-rate based on instantaneous CSI \cite{zhou2020intelligent,guo2020weighted,pan2020multicell}, the approximate sum-rate in (\ref{R_k_asy}) only relies on statistical CSI. 
Accordingly, the reformulated problem is given as
\begin{equation} \label{problem}
\begin{aligned}
\mathcal{P}\,:\,&\max\limits_{\mathbf{W},\boldsymbol{\theta}}\quad \sum\limits_{k = 1}^K{\log _2}\det \bigg[ {{\bf{I}}_L}+ \mathbf{\Phi}_k{ \big(\sum\limits_{i = 1,i\neq k}^K\mathbf{\Phi}_i+
{\frac{{\sigma_k^2}}{\varphi_k{\mathrm{Tr}\left(\mathbf{T}_I\mathbf{\Theta}\mathbf{R}_I\mathbf{\Theta}^{\mathrm{H}}\right)}}\mathbf{I}_L} \big)}^{ - 1} \bigg], \\
& \begin{array}{r@{\quad}r@{}l@{\quad}l}
\mathrm{s.t.}
&(\mathrm{C1}):&\mathrm{Tr}(\mathbf{W}\mathbf{W}^\mathrm{H})\leq P_{max},\\
&(\mathrm{C2}):&0\leq\theta_n\leq 2\pi, n=1,\cdots,N.
\end{array}
\end{aligned}
\end{equation}

By observing the objective function in (\ref{problem}), we can find that the transmit beamforming matrix $\mathbf{W}$ and phase shift  vector $\boldsymbol{\theta}$ are only related to $\mathbf{\Phi}_j,j\in\{1,\cdots,K\}$ and ${\mathrm{Tr}\left(\mathbf{T}_I\mathbf{\Theta}\mathbf{R}_I\mathbf{\Theta}^\mathrm{H}\right)}$, respectively. Besides, (C1) and (C2) are the constraints for two optimization variables, respectively. Unlike the previous works \cite{wu2019intelligent,wu2020beamforming,zhou2020intelligent,
guo2020weighted,pan2020multicell,zhang2020sum,ye2020joint,
pan2020intelligent,zheng2020intelligent,mu2020exploiting,
zhao2020performance,wang2020intelligent,li2020reconfigurable,
zhu2020Stochastic,han2019large,zhang2020transmitter,wang2021joint,
chaaban2020opportunistic,kammoun2020asymptotic,alwazani2020intelligent}, it can be found that the designs of transmit beamforming matrix $\mathbf{W}$ and RIS phase shift vector $\boldsymbol{\theta}$ can be decoupled in the reformulated problem. In what follows, the optimizations of $\mathbf{W}$ and $\boldsymbol{\theta}$ are tackled one by one.
\subsection{Optimization of Phase Shifts} \label{theta_design}
Noticing that $\boldsymbol{\theta}$ is only related to the trace term ${\mathrm{Tr}\left(\mathbf{T}_I\mathbf{\Theta}\mathbf{R}_I\mathbf{\Theta}^\mathrm{H}\right)}$ in the denominator of the inverse matrix in (\ref{problem}), and the objective function is monotonically increasing w.r.t. ${\mathrm{Tr}\left(\mathbf{T}_I\mathbf{\Theta}\mathbf{R}_I\mathbf{\Theta}^\mathrm{H}\right)}$, the optimal phase shift design can be transformed to the maximization of the trace term ${\mathrm{Tr}\left(\mathbf{T}_I\mathbf{\Theta}\mathbf{R}_I\mathbf{\Theta}^\mathrm{H}\right)}$ under the constraint of (C2). Then the corresponding optimization problem w.r.t. $\boldsymbol{\theta}$ can be casted as
\begin{equation}
\begin{aligned}
 \label{P_1}
\mathcal{P}1\,:\,&\max\limits_{\boldsymbol{\theta}}\quad {\mathrm{Tr}\left(\mathbf{T}_I\mathbf{\Theta}\mathbf{R}_I\mathbf{\Theta}^\mathrm{H}\right)}, \\
& \begin{array}{r@{\quad}r@{}l@{\quad}l}
\mathrm{s.t.}
&(\mathrm{C2}):&\,0\leq\theta_n\leq 2\pi, n=1,\cdots,N.
\end{array}
\end{aligned}
\end{equation}
Let $\mathbf{v}=[e^{j\theta_1},\cdots,e^{j\theta_n},\cdots,e^{j\theta_N}]^{\mathrm{T}}$, together with the matrix identity of \cite[Lemma 7.5.2]{horn2012matrix},   problem (\ref{P_1}) can be reduced to
\begin{equation}\label{SDP}
\begin{aligned}
\mathcal{P}1\,:\,&\max\limits_{\mathbf{V}}\quad \mathrm{tr}\left(\mathbf{\Xi}\mathbf{V}\right), \\
& \begin{array}{r@{\quad}r@{}l@{\quad}l}
\mathrm{s.t.}
&&\mathbf{V}_{n,n}=1, n=1,\cdots,N,\\
&&\mathbf{V}\succeq0 \& \mathrm{rank}(\mathbf{V})=1,
\end{array}
\end{aligned}
\end{equation}
where $\mathbf{\Xi}=\mathbf{T}_I\odot\mathbf{R}_I^{\mathrm{T}}$, $\mathbf{V}=\mathbf{v}\mathbf{v}^\mathrm{H}$ and $\mathbf{T}_I\odot\mathbf{R}_I^{\mathrm{T}}$ denotes the Hadamard product of $\mathbf{T}_I$ and $\mathbf{R}_I$. Since both $\mathbf{T}_I$ and $\mathbf{R}_I$ are positive semi-definite Hermitian matrices, $\mathbf{\Xi}$ is also a positive semi-definite Hermitian matrix \cite[P477]{horn2012matrix}.

It is well known that (\ref{SDP}) is a NP-hard quadratically constrained quadratic programs(QCQP) problem due to the rank-one constraint. There are some popular tools that can deal with this problem, such as semi-definite relaxation (SDR) algorithm\cite{wu2019intelligent}, majorization-minimization (MM) algorithm \cite{pan2020multicell},  the iterative
algorithm for continuous phase case (IA-CPC) \cite{cui2017quadratic}, wideband beampattern formation via iterative techniques (WBFIT) \cite{he2010wideband}, etc.
Although, the  SDR algorithm has been widely adopted to solve the  NP-hard QCQP problem with rank one constraint, it relies on the convex optimization solvers such as CVX, and requires a  sufficiently large number of randomizations to guarantee the approximation with high computational complexity \cite{so2007approximating}.

In order to reduce the computational complexity for large phase shift elements, we apply the PGA algorithm to obtain a sub-optimal solution  $\boldsymbol{\theta}^\star$ in this work. The main idea of PGA algorithm is similar to the concept of gradient descent technique by employing gradient ascent to monotonically increase the objective function of (\ref{P_1}). However, at each iteration, the solution should  project onto the closest feasible point that satisfies the unit-modulus constraint.  The main steps of PGA algorithm in each iteration $(t)$ is given as following.

Firstly, we have to compute the gradient ascent direction. By taking the first-order derivative of $\varphi=\mathrm{Tr}\left(\mathbf{T}_I\mathbf{\Theta}\mathbf{R}_I\mathbf{\Theta}^\mathrm{H}\right)$ w.r.t. $v_n=e^{j\theta_n}$, we have
\begin{equation}\label{gradient}
\frac{\partial \varphi}{\partial v_n}=\mathrm{Tr}(\mathbf{T}_I(\mathbf{E}_{nn}\mathbf{R}_I\mathbf{\Theta}^\mathrm{H}-v_n^{-1}\mathbf{\Theta}\mathbf{R}_I\mathbf{E}_{nn})),
\end{equation}
where $\mathbf{E}_{nn}$ denotes the all-zero $N\times N$ matrix except that the $(n,n)$-th element is 1.

 We denote $\mathbf{v}^{(t)}=[v_1^{(t)},\cdots,v_N^{(t)}]$, where $|v_n^{(t)}|=1,n\in\{1,\cdots,N\}$. Then, the new  vector $\tilde{\mathbf{v}}^{(t+1)}$ is updated by setting $\mathbf{p}^{(t)}=[\frac{\partial \varphi}{\partial v_1},\cdots,\frac{\partial \varphi}{\partial v_N}]$ as ascent direction,
\begin{equation}
\tilde{\mathbf{v}}^{(t+1)}=\mathbf{v}^{(t)}+\alpha \mathbf{p}^{(t)},
\end{equation}
 where $\alpha$ is the step size of gradient ascent.

 It is obviously found
that the update vector $\tilde{\mathbf{v}}^{(t+1)}$ does not satisfy the unit-modulus constraint in (\ref{SDP}).  Hence, it has to be projected onto the closest feasible point that satisfies the constraint, and the corresponding phase shift vector is given as
 \begin{equation}
\mathbf{v}^{(t+1)}= \mathrm{exp}(j\mathrm{arg}(\tilde{\mathbf{v}}^{(t+1)})).
\end{equation}
  Since gradient ascent is along the monotonically increasing direction of $\varphi$ at each iteration and $\mathrm{tr}\left(\mathbf{\Xi}\mathbf{v}^{(t)}{\mathbf{v}^{(t)}}^{\mathrm{H}}\right)$ is upper bounded, the convergence of the PGA algorithm is therefore guaranteed.  When the algorithm converges, we can obtain the sub-optimal phase shift vector as $\boldsymbol{\theta}^\star=\mathrm{arg}(\mathbf{v}^{(t)})$. It is worth emphasizing that global optimality of the phase shifts cannot be guaranteed since the problem (\ref{P_1}) is not convex w.r.t. $\boldsymbol{\theta}$.
\subsection{Optimization of Transmit beamforming Matrix $\mathbf{W}$}\label{Beamforming_design}
Next, the transmit beamforming matrix $\mathbf{W}$ at the BS will be optimally designed to maximize the approximate sum-rate of RIS-aided multi-user systems. Since the phase shifts are independent of $\mathbf{W}$ and have been designed in Section \ref{theta_design}, the optimization problem for $\mathbf{W}$ can be simplified as
\begin{equation}
\begin{aligned} \label{P_bar}
\mathcal{P}2\,:\,&\max\limits_{\mathbf{W}}\quad \sum\limits_{k = 1}^K {\log _2}\det\bigg[ {{\bf{I}}_L} + \mathrm{Tr}\big(\mathbf{T}_B {\mathbf{W}_k}{\mathbf{W}_k^\mathrm{H}}\big) \tilde{\mathbf{R}}_{U_{k}}\times{\bigg(\sum\limits_{i = 1,i\neq k}^K\mathrm{Tr}\big(\mathbf{T}_B {\mathbf{W}_i}{\mathbf{W}_i^\mathrm{H}}\big) \tilde{\mathbf{R}}_{U_{k}}+\epsilon_k \mathbf{I}_L\bigg)}^{ - 1} \bigg], \\
& \begin{array}{r@{\quad}r@{}l@{\quad}l}
s.t.
&(\mathrm{C1}):&\mathrm{Tr}(\mathbf{W}\mathbf{W}^\mathrm{H})\leq P_{max},\\
\end{array}
\end{aligned}
\end{equation}
where  $\tilde{\mathbf{R}}_{U_{k}}={\mathbf{R}}_{U_{k}}/N$ and $\epsilon_k=\frac{{\sigma_k^2}}{\varphi_k{\mathrm{Tr}\left(\mathbf{T}_I\mathbf{\Theta}^\star\mathbf{R}_I{\mathbf{\Theta}^\star}^\mathrm{H}\right)}}$ is deterministic for a given optimal phase shift vector  $\boldsymbol{\theta}^\star$, as shown in Section \ref{theta_design}.

This problem is also non-convex since the source transmit beamformer $\mathbf{W}$ appears
in both the numerator and the denominator of the approximate sum-rate.  In the following we provide a   globally optimal solution  for transmit beamforming matrix by using the water-filling technique and  Lagrangian multiplier method.
At first, by introducing auxiliary variables $\bar{P}$ and $\boldsymbol{\eta}=[\eta_1,\cdots,\eta_K]^\mathrm{T},\,(0\leq \eta_k\leq1, \sum\limits_{k = 1}^K\eta_k=1)$, where $\bar{P}={\mathrm{Tr}}(\mathbf{T}_B\mathbf{W}\mathbf{W}^\mathrm{H})$, and $\boldsymbol{\eta}$ denotes the power weight vector for each user precoder after multiplying by the BS antenna correlation matrix $\mathbf{T}_B$, i.e., $\eta_k\bar{P}=\mathrm{Tr}\left(\mathbf{T}_B {\mathbf{W}_k}{\mathbf{W}_k^\mathrm{H}}\right)$, the optimization problem $\mathcal{P}_2$ can be reformulated as
\begin{equation} \label{P_2}
\begin{aligned}
&\max\limits_{\bar{P},\boldsymbol{\eta},\mathbf{W}}\quad \sum\limits_{k = 1}^K {\log _2}\det\big[ {{\bf{I}}_L} + \eta_k\bar{P} \tilde{\mathbf{R}}_{U_{k}}{\big((1-\eta_k)\bar{P} \tilde{\mathbf{R}}_{U_{k}}+\epsilon_k \mathbf{I}_L\big)}^{ - 1} \big], \\
& \begin{array}{r@{\quad}r@{}l@{\quad}l}
s.t.
&(\mathrm{C1}):&\mathrm{Tr}(\mathbf{W}\mathbf{W}^\mathrm{H})\leq P_{max},\\
&(\mathrm{C3}):&\,\,\bar{P}={\mathrm{Tr}}(\mathbf{T}_B\mathbf{W}\mathbf{W}^\mathrm{H}),\\
&(\mathrm{C4}):&\sum\limits_{k = 1}^K\eta_k=1, \eta_k\geq0.\\
\end{array}
\end{aligned}
\end{equation}
The objective function in (\ref{P_2}) can be rewritten as
\begin{equation}\label{f_bar_P}
\begin{aligned}
f(\bar{P},\boldsymbol{\eta})&=\sum\limits_{k = 1}^K {\log _2}\det\big[\big(\epsilon_k{\bf{I}}_L+ \bar{P}\tilde{\mathbf{R}}_{U_{k}}\big)\times{\big[(1-{\eta_k})\bar{P} \tilde{\mathbf{R}}_{U_{k}}+\epsilon_k{\bf{I}}_L\big]}^{ - 1} \big]\\
&=\sum\limits_{k = 1}^K\sum\limits_{l = 1}^L  {\log _2}\frac{\epsilon_k+r_{k,l}\bar{P}}{\epsilon_k+(1-\eta_k)r_{k,l}\bar{P}},
\end{aligned}
\end{equation}
where $r_{k,l}=\lambda_{U_{k_l}}/N$ is the $l$-th eigenvalue of matrix $\tilde{\mathbf{R}}_{U_k}$.
It is readily found from (\ref{f_bar_P}) that $f(\bar{P},\boldsymbol{\eta})$
is a monotonically increasing function of $\bar{P}$, because the first derivative of ${\log _2}\frac{\epsilon_k+r_{k,l}\bar{P}}{\epsilon_k+(1-\eta_k)r_{k,l}\bar{P}}$ w.r.t. $\bar{P}$ is always greater than $0$. Therefore, $\mathbf{W}$ should be optimally chosen to maximize $\bar{P}$, such that
\begin{equation}
\begin{aligned}
 \label{P_3}
\mathcal{P}3\,:\,&\max\limits_{\mathbf{W}}\quad {\mathrm{Tr}}(\mathbf{T}_B\mathbf{W}\mathbf{W}^\mathrm{H}), \\
& \begin{array}{r@{\quad}r@{}l@{\quad}l}
s.t.
&(\mathrm{C1}):&\mathrm{Tr}(\mathbf{W}\mathbf{W}^\mathrm{H})\leq P_{max},\\
&(\mathrm{C5}):&\mathrm{Rank}(\mathbf{W})= Kd.

\end{array}
\end{aligned}
\end{equation}
where the constraint (C5) denotes the $Kd$ data steams at BS.
It is easily  observed that the above problem is equivalent to  Rayleigh quotient problem \cite{gong2018secure}. As such, we present the optimal structure of $\mathbf{W}$ to maximize $\bar{P}$ in the following  theorem.
 \begin{theorem}
 Assuming that the EVD of $\mathbf{T}_B$ is $\mathbf{T}_B=\mathbf{U}_{\mathbf{T}}\mathbf{\Sigma}_{\mathbf{T}}\mathbf{U}_{\mathbf{T}}^{\mathrm{H}}$, where $\{\mathbf{U}_\mathbf{T},\mathbf{\Sigma}_\mathbf{T}\}\in\mathbb{C}^{M\times M}$, $\mathbf{\Sigma}_{\mathbf{T}}$  denotes a  diagonal matrix and its diagonal elements
 are the eigenvalues of $\mathbf{T}_B$ with descending ordered, i.e., $\mathbf{\Sigma}_{\mathbf{T}}(1,1)\geq\mathbf{\Sigma}_{\mathbf{T}}(2,2)\geq\mathbf{\Sigma}_{\mathbf{T}}(M,M)$.   The singular value decomposition (SVD) of $\mathbf{W}$ is $\mathbf{W}=\mathbf{U}_\mathbf{W}\mathbf{\Sigma}_\mathbf{W}^{\frac{1}{2}}\mathbf{V}_\mathbf{W}^{\mathrm{H}}$, where $\mathbf{U}_\mathbf{W}\in\mathbb{C}^{M\times M}$, $\mathbf{\Sigma}_\mathbf{W}^{\frac{1}{2}}\in\mathbb{C}^{M\times Kd}$,
 $\mathbf{V}_\mathbf{W}^{\mathrm{H}}\in\mathbb{C}^{Kd\times Kd}$. We can obtain that the optimal solution for $\mathbf{W}$ is the one which satisfies
 \begin{equation}\label{W_ast}
 \mathbf{W}^\star=\mathbf{U}_{\mathbf{T}}\mathbf{\Sigma}_{\mathbf{W}^\star}^{\frac{1}{2}}\mathbf{V}_\mathbf{W}^{\mathrm{H}}, \end{equation}
 where $\mathbf{\Sigma}_{\mathbf{W}^\star}$ denotes the all-zero $M\times Kd$ matrix except that the entry of the $(1,1)$-th element is $P_{max}$, $\mathbf{V}_\mathbf{W}^{\mathrm{H}}$ is the arbitrary unitary matrix.
 \end{theorem}

By substituting (\ref{W_ast}) into  (\ref{P_3}), we can obtain the maximum $\bar{P}$ as $\bar{P}^\star=t_{max}P_{max}$, where $t_{max}$ is the maximum eigenvalue of $\mathbf{T}_B$.

We then put our focus on optimizing the power weight vector $\boldsymbol{\eta}$. By substituting $\bar{P}^\star$ into (\ref{P_2}) and ignoring the constant terms, the asymptotic sum-rate maximization w.r.t. $\boldsymbol{\eta}$ is reformulated as
\begin{equation}\label{beta}
\begin{aligned}
\,&\min\limits_{\boldsymbol{\eta}}\sum\limits_{k = 1}^K {\ln}\det\left[\epsilon_k{\bf{I}}_L+(1-{\eta_k})\bar{P}^\star \tilde{\mathbf{R}}_{U_{k}}\right], \\
& \begin{array}{r@{\quad}r@{}l@{\quad}l}
s.t.
&(\mathrm{C4}):&\sum\limits_{k = 1}^K\eta_k=1, \eta_k\geq0.\\
\end{array}
\end{aligned}
\end{equation}
We can find that (\ref{beta}) is a water-filling problem, and the corresponding Lagrangian function is given by
\begin{equation}
\mathcal{L}(\boldsymbol{\eta},\tau)=\sum\limits_{k = 1}^K \sum\limits_{l = 1}^L{\ln}{\left[\epsilon_k+(1-{\eta_k}) r_{k,l}\bar{P}^\star\right]}+\tau\big(\sum\limits_{k = 1}^K\eta_k-1\big),
\end{equation}
where $\tau$ is the lagrange multipliers associated with constraint $\sum\limits_{k = 1}^K\eta_k=1$.  Due to complementary slackness, the optimal lagrange multipliers are zeros in constraints $\eta_k\geq0$.

 By taking the derivative of $\mathcal{L}(\boldsymbol{\eta},\tau)$ w.r.t. $\eta_k$, we have
 \begin{equation}\label{first_order}
\frac{\partial \mathcal{L}(\boldsymbol{\eta},\tau)}{\partial \eta_k}=\tau-\sum\limits_{l = 1}^L\frac{1}{1-{\eta_k}+\epsilon_k/( r_{k,l}\bar{P}^\star)}).
 \end{equation}

Note that by setting $\frac{\partial \mathcal{L}(\boldsymbol{\eta},\tau)}{\partial \eta_k}=0$, we can obtain the optimal solution of $\boldsymbol{\eta}$. However, according to the Galois theory \cite{cox2011galois}, there is no closed-from solution for the equation $\frac{\partial \mathcal{L}(\boldsymbol{\eta},\tau)}{\partial \eta_k}=0$ w.r.t. $\eta_k$ if $L\geq5$. 
To deal with this issue, we have one proposition for the optimal solution $\boldsymbol{\eta}^*$ as follows.
\begin{Proposition}
By setting (\ref{first_order}) equal to zero, the optimal solution $\boldsymbol{\eta}^\star$  can be obtained by using a nested bisection method, with which the globally optimal solution can be obtained.
\end{Proposition}
\begin{proof}
Since $\eta_k\in [0,\,1]$, we can easily obtain the upper and lower bounds
of $\tau$, which are respectively given by
 \begin{equation}\label{tau_up}
  \tau<\max\limits_{k\in{1,\cdots,K}}\sum\limits_{l = 1}^L\frac{1}{\epsilon_k/( r_{k,l}\bar{P}^\star)}\triangleq\tau^{up},
  \end{equation}
  and
    \begin{equation}\label{tau_lp}
\tau^{lp}\triangleq\min\limits_{k\in{1,\cdots,K}}\sum\limits_{l = 1}^L\frac{1}{1+\epsilon_k/(r_{k,l}\bar{P}^\star )}< \tau.
  \end{equation}
Given an initial value $\tau^{(0)}$, we can always obtain a solutions for each $\eta^{(i)}_k$ by using nested bisection method, where $i$ is the number of the iterations and $k\in\{1,\cdots,K\}$. We set $\tau^{ub}=\tau^{i}$ if $\sum\nolimits_{k =1}^K\eta^{(i)}_k<1$, and $\tau^{lb}=\tau^{i}$ otherwise. Then $\tau^{(i)}$ is updated as $\tau^{(i)}=\frac{\tau^{lp}+\tau^{up}}{2}$. With the updated $\tau$, $\eta^{(i+1)}_k$ is recalculated. This process is repeated until $\sum\nolimits_{k =1}^K\eta^{(i+1)}_k-1$ less than the error tolerance $\varepsilon$. Since the first-order derivative function $\frac{\partial \mathcal{L}(\boldsymbol{\eta},\tau)}{\partial \eta_k}$ is a monotonically decreasing and increasing w.r.t. $\eta_k$ and $\tau$, respectively, we can guarantee that the nested bisection method converges to the unique globally optimal solution. The nested bisection method is summarized in Algorithm \ref{alg} in next subsection.

\end{proof}

Based on the optimal solution $\bar{P}^\star=t_{max}P_{max}$ and the optimal power weight vector $\boldsymbol{\eta}^\star$, we can obtain the optimal solution for $\mathbf{W}_k, k\in[1,\cdots,K]$ as
  \begin{equation}\label{opt_W_k}
 \mathbf{W}_k^\star=\mathbf{U}_{\mathbf{T}}\mathbf{\Sigma}_{\mathbf{W}_k^\star}^{\frac{1}{2}}\mathbf{V}_{\mathbf{W}_k}^{\mathrm{H}},
  \end{equation}
  where $\mathbf{\Sigma}_{\mathbf{W}_k^\star}$ denotes the all-zero $M\times d$ matrix except that the $(1,1)$-th element is $\eta_k^\star P_{max}$, $\mathbf{V}_{\mathbf{W}_k}\in\mathcal{C}^{d\times d}$ is an arbitrary unitary matrix.

  Finally, the optimal design method for $\mathbf{W}_k$ can be summarized as follows. One one hand, the left-singular vectors of $\mathbf{W}_k$ should be consistent with eigenvectors of $\mathbf{T}_B$, and the maximum singular value of   $\mathbf{W}_k$ corresponds to the dominant eigenvalue of $\mathbf{T}_B$ to maximize ${\mathrm{Tr}}(\mathbf{T}_B\mathbf{W}\mathbf{W}^\mathrm{H})$. On the other hand, the power allocation among $\{\mathbf{W}_1,\cdots,\mathbf{W}_K\}$ for each user should be designed on the basis of the statistical CSI $\mathbf{R}_{U_k}$, including path-loss for cascaded channel, the second-order statistics of small-scale fading for each channel, and the spatial correlation matrix at each user, which  can be solved by applying the bisection search.

\subsection{Summary and Disscussion}
In Sections \ref{theta_design} and \ref{Beamforming_design}, we solve the optimization problem by decoupling it into two sub-problems for the RIS-aided multi-user MIMO system. Clearly from (\ref{opt_W_k}), the direction of the optimal transmit beamforming $\mathbf{W}_k$ is designed to be aligned with the transmit correlation $\mathbf{T}_B$ at the BS, while the power allocation for each user is determined by the receive correlation at the user itself. Meanwhile, it is shown in (\ref{P_1}) that the phase shifts $\boldsymbol{\theta}$ are designed purely by the spatial correlations at the RIS. The pseudo-code of the joint optimization of the precoding matrix $\mathbf{W}$ and RIS phase shift vector $\boldsymbol{\theta}$ is provided in Algorithm \ref{alg}.

\begin{algorithm}
  \caption{: Joint optimization of the precoding matrix $\mathbf{W}$ and RIS phase shift vector $\boldsymbol{\theta}$
  }\label{alg}
\begin{algorithmic}[1] 
\STATE Initalize        \begin{itemize}
                   \item The feasible solution $\boldsymbol{\theta}^{(1)}$;
                   \item The iteration number $t=1$, and the accuracy $\varepsilon$;
                 \end{itemize}
\STATE Calculate  the value of the OF in (\ref{SDP}) as $\varphi(\mathbf{v}^{(t)})$;
\REPEAT
\STATE Calculate  the gradient $\mathbf{p}^{(t)}=[\frac{\partial \varphi}{\partial v_1},\cdots,\frac{\partial \varphi}{\partial v_N}]$, where $\frac{\partial \varphi}{\partial v_n}$ is given in (\ref{gradient});
\STATE Calculate $\tilde{\mathbf{v}}^{(t+1)}=\mathbf{v}^{(t)}+\alpha \mathbf{p}^{(t)}$ for gradient  ascent direction;
\STATE Update $\mathbf{v}^{(t+1)}= \mathrm{exp}(j\mathrm{arg}(\mathbf{v}^{(t+1)}))$;
\STATE Calculate  the value of the OF in (\ref{SDP}) as $\varphi(\mathbf{v}^{(t+1)})$;
\UNTIL{$|\varphi(\mathbf{v}^{(t+1)})-\varphi(\mathbf{v}^{(t)})|\leq \varepsilon$};

\STATE Calculate $\mathbf{U}_{\mathbf{T}}$ according to eigenvalue descending order, given $\mathbf{T}_B$;
\STATE Initalize        \begin{itemize}
                   \item The upper and lower bounds on $\tau$, $\tau^{ub}$ in (\ref{tau_up}) and $\tau^{lb}$ in (\ref{tau_lp});
                   \item The upper and lower bounds on $\eta_k$, $\eta_k^{lb}=0$ and $\eta_k^{ub}=1$ ($k\in \{0,\cdots,K\}$);
                   \item The user index $k=0$;
                 \end{itemize}
\REPEAT
\STATE $\tau=(\tau^{lb}+\tau^{lb})/2$;
\REPEAT
\STATE $k=k+1$;
\STATE Calculate $\eta_k^\star$ with  bisection search method basing on (\ref{first_order});
\UNTIL{$k=K$}
\STATE \textbf{if} $\sum^K_{k=1}\eta_k-1\leq0$ \textbf{then} $\tau^{ub}=\tau$;
\STATE \textbf{else} $\tau^{lb}=\tau$;
\UNTIL{$|\sum\limits_{k = 1}^K\eta_k-1|\leq \varepsilon$};
\STATE Calculate $\mathbf{W}_k^\star=\mathbf{U}_{\mathbf{T}}\mathbf{\Sigma}_{\mathbf{W}_k^\star}^{\frac{1}{2}}\mathbf{V}_{\mathbf{W}_k}^{\mathrm{H}}$, $\mathbf{\Sigma}_{\mathbf{W}_k^\star}{(1,1)}=\eta_k^\star P_{max;}$
\ENSURE The optimal phase shift vector $\boldsymbol{\theta}^\star=\mathrm{arg}(\mathbf{v}^{(t)})$ and the precoder $\mathbf{W}^\star=[\mathbf{W}_1^\star,\cdots,\mathbf{W}_K^\star]$. 
\end{algorithmic}
\end{algorithm}
Next, we investigate the computational complexity of  Algorithm \ref{alg}. The main complexity of Algorithm \ref{alg} originates from Steps 3 and 9. Specifically, in Step 3, the complexity of the gradient calculation  is given by $\mathcal{O}(tN^{2})$, where $t$ denotes the iteration number required until converge. In Step 9,
the complexity of the EVD is in the order of $\mathcal{O}(M^{3})$. Therefore, the overall complexity of Algorithm \ref{alg} is $\mathcal{O}(max[tN^{2},M^{3}])$. By employing the BCD algorithm in \cite{pan2020multicell}, the overall computational complexity for the sum-rate maximization based on instantaneous CSI is $\mathcal{O}(t_{BCD}(max[KM^{3},LM^{2},N^3+t_{MM}N^2]))$, where $t_{BCD}$ and $t_{MM}$ represent the number of iterations required until convergence for the BCD algorithm and Majorization-Minimization (MM) algorithm for phase shift design, respectively.  Since the number of RIS elements $N$ is usually much larger than the number of BS antennas $M$, the complexity of the proposed algorithm is much lower than that of  \cite{pan2020multicell}. Moreover,  the communications overhead in the proposed algorithm is also lower due to the avoidance of frequent channel estimation and frequent adjustments of transmit beamforming and phase shifts based on statistical CSI.

\section{Simulation Results}\label{section_V}
In this section, numerical results are presented to validate our foregoing analysis and examine the sum-rate performance achieved by our proposed algorithm. For illustrative purpose, we adopt the exponential correlation model to characterize the spatial correlation among antennas and reflecting elements, and the $(i,j)$-th element of correlation matrix $\mathbf{A}$ is $\mathbf{A}(i,j)=\rho_{\mathbf{A}}^{|i-j|},\mathbf{A}\in \{\mathbf{T}_B,\mathbf{R}_I,\mathbf{T}_I,\mathbf{R}_{U_k}\}$, where $0\leq \rho_{\mathbf{A}}\leq1$ denotes the correlation coefficient between any two adjacent antenna elements \cite{feng2017power,feng2018impact}. Furthermore, the large-scale path loss is modeled as $\beta_\mathbf{C}=({d_{\mathbf{C}}}/{d_0})^{-\alpha_\mathbf{C}}$, where $d_{\mathbf{C}}$ and ${\alpha_\mathbf{C}}$ denote the communication distance and path loss exponent of link $\mathbf{C}$, respectively. $d_0=10\,\mathrm{m}$ stands for a
reference distance \cite{huan2020rate}.  Unless otherwise
mentioned, the system setup parameters are set as $M=10$, $K=2$, $d=3$, $L=3$, $d_{\mathbf{H}_{BI}}=10\,\mathrm{m}$, $d_{\mathbf{H}_{IU_1}}=d_{\mathbf{H}_{IU_2}}=30\,\mathrm{m}$, $\alpha_{\mathbf{H}_{BI}}=2.2$, $\alpha_{\mathbf{H}_{IU_1}}=\alpha_{\mathbf{H}_{IU_2}}=3$,  $P_{max}=40\mathrm{dBmW}$.
\subsection{Verifications}
\begin{figure}[!htb]
\centering
\includegraphics[width=3.5in]{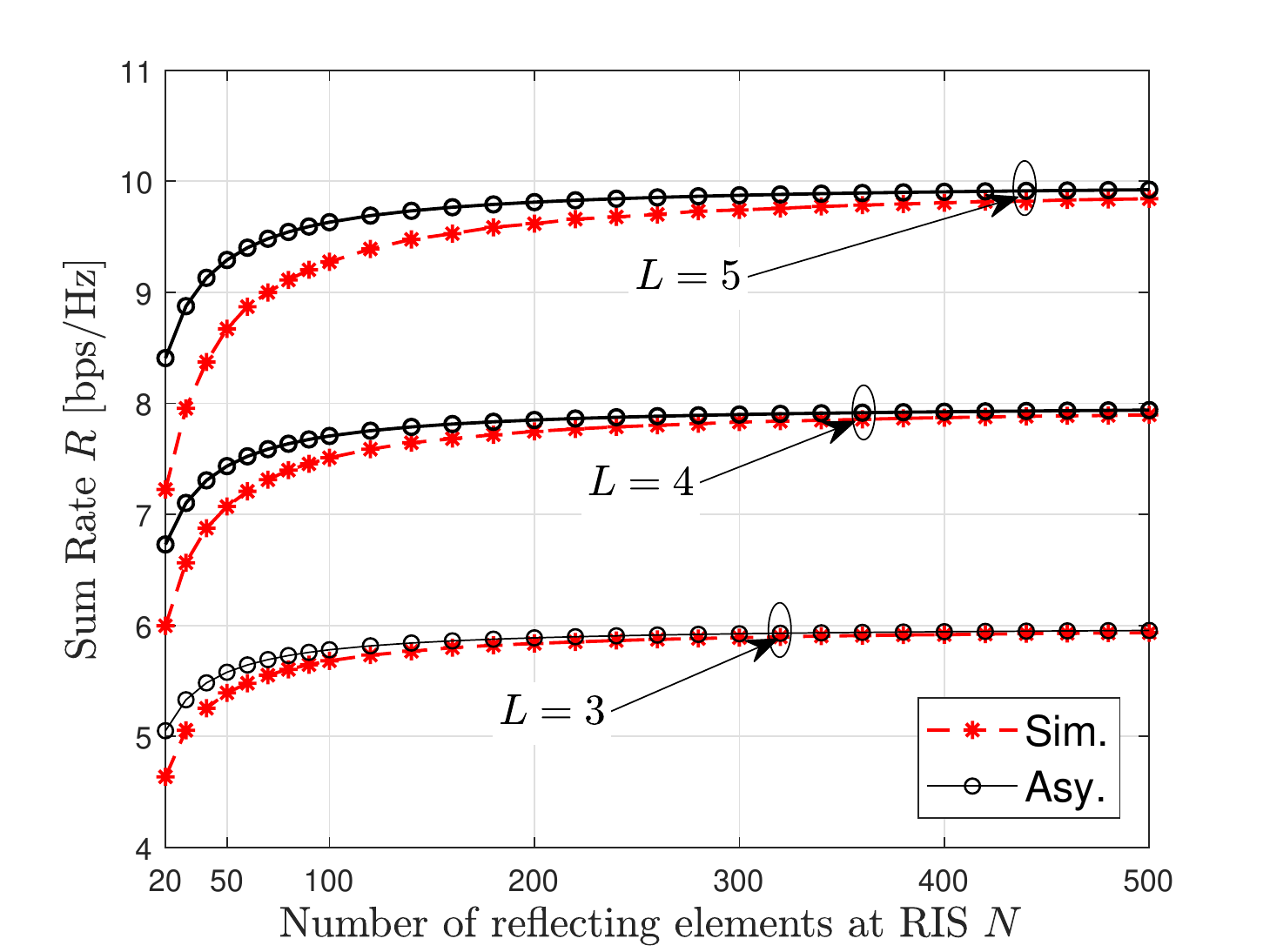}
\caption{The sum-rate versus the number of reflecting elements $N$ under different number of antennas at the user side.}\label{figure2}
\end{figure}
Fig. \ref{figure2} depicts the simulated (labeled as ``Sim.'') and asymptotic (labeled as ``Asy.'') sum-rate
under different number of user antennas $L$ as the reflecting elements number $N$ increases, wherein EPA precoder and exponential correlation model are employed by setting $\rho_\mathbf{A}=0.2$. The phase shift for each element is set as $\theta_n=\frac{\pi}{3}$. It can be seen from this figure that the simulation result approaches to the asymptotic sum-rate as the number of reflecting element increases, and coincides well with the asymptotic sum-rate for over 200 reflecting elements when $L=3$. This confirms the correctness of the deterministic equivalent analysis in Section III. Fig. \ref{figure2} also shows that the asymptotic sum-rate increases with $L$  due to the benefit of spatial diversity gain for multi-antenna receiver. In the meantime, it can be observed from this figure that the sum-rate for RIS-aided MIMO system is upper bounded. The upper bound of the sum-rate can be determined by combining (\ref{R_k_asy}) and \cite[Lemma 2]{feng2019two}.
\subsection{Impacts of Spatial Correlation}
\begin{figure}[!htb]
\centering
\includegraphics[width=3.5in]{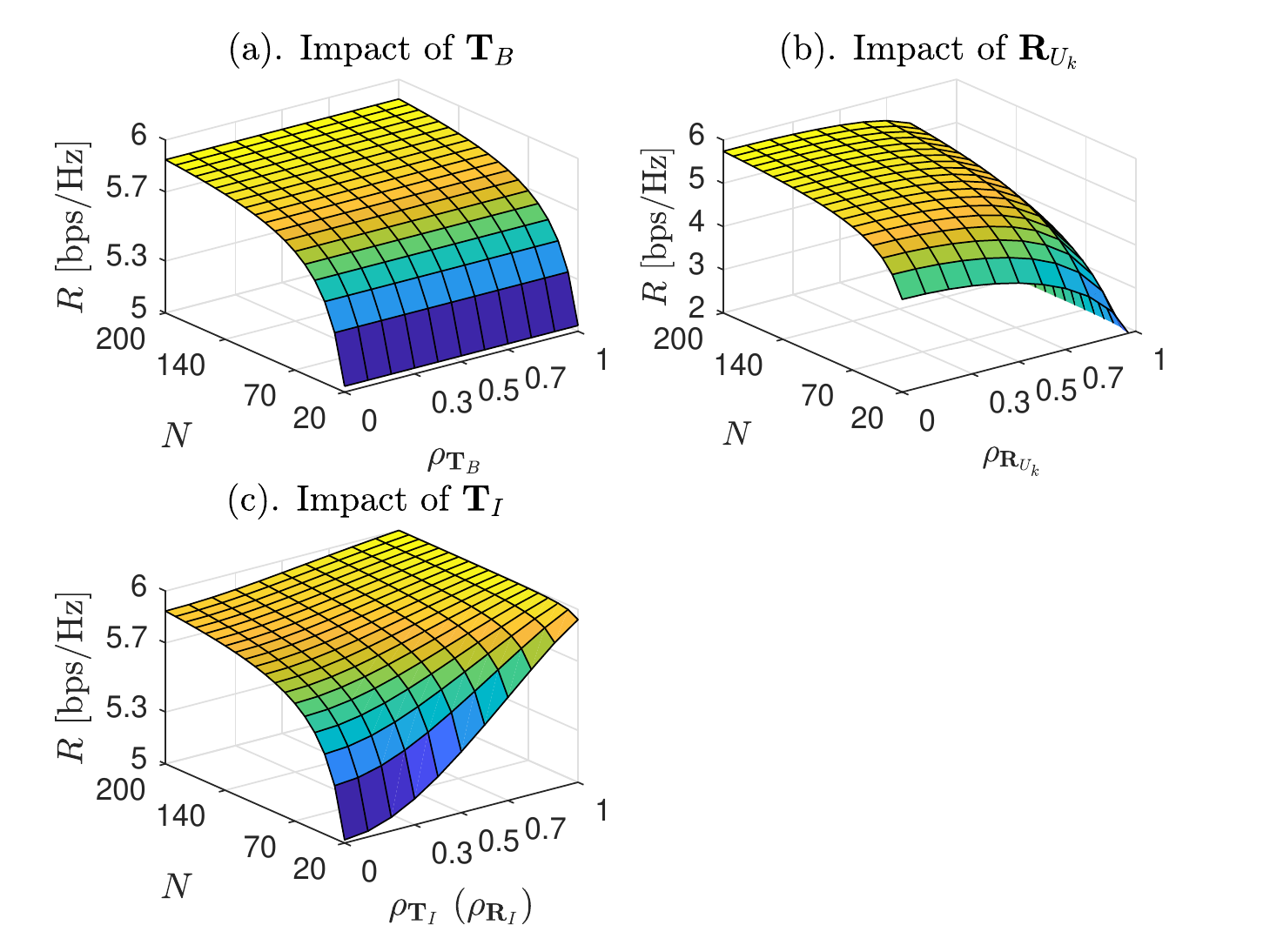}
\caption{The sum-rate versus the number of reflecting elements $N$ and correlation coefficient $\rho_{\mathbf{A}}$. }\label{figure3}
\end{figure}
The effect of spatial correlation on each terminal is examined in Fig. \ref{figure3}. It is worth mentioning that all the exponential correlation coefficients are set as $0.1$ unless otherwise noticed. Notice that in Fig. \ref{figure3}(b), we set $d_{\mathbf{H}_{IU_1}}=d_{\mathbf{H}_{IU_2}}=40 \mathrm{m}$ to study the impact of spatial correlation on user sides, in Fig. \ref{figure3}(c) we set $\rho_{\mathbf{T}_I}=\rho_{\mathbf{R}_I}$ to verify the impact of spatial correlation on RIS.  It can be observed from Figs. \ref{figure3}(a)-(c) that the increase of the number of reflecting elements causes the enhancement of the system sum-rate. We can see from Fig. \ref{figure3}(a) that the antenna correlation at the BS sides $\mathbf{T}_B$ is irrelevant to the asymptotic sum-rate. However, the spatial correlation at user $\mathbf{R}_{U_k}$ negatively influences the sum-rate as observed in Fig. \ref{figure3}(b). In contrast to the spatial correlation at user side, it is readily seen in Fig. \ref{figure3}(c) that the spatial correlation at RIS has a positive effect on the asymptotic sum-rate.
The numerical results are consistent with the analytical results in Section \ref{correlation_impact}.
\subsection{Optimal System Design}
In this subsection, we present numerical results to examine the sum-rate achieved by our proposed algorithm. We compare the performance of the proposed algorithm with other benchmarking schemes.
 \begin{figure}[!htb]
\centering
\includegraphics[width=3.5in]{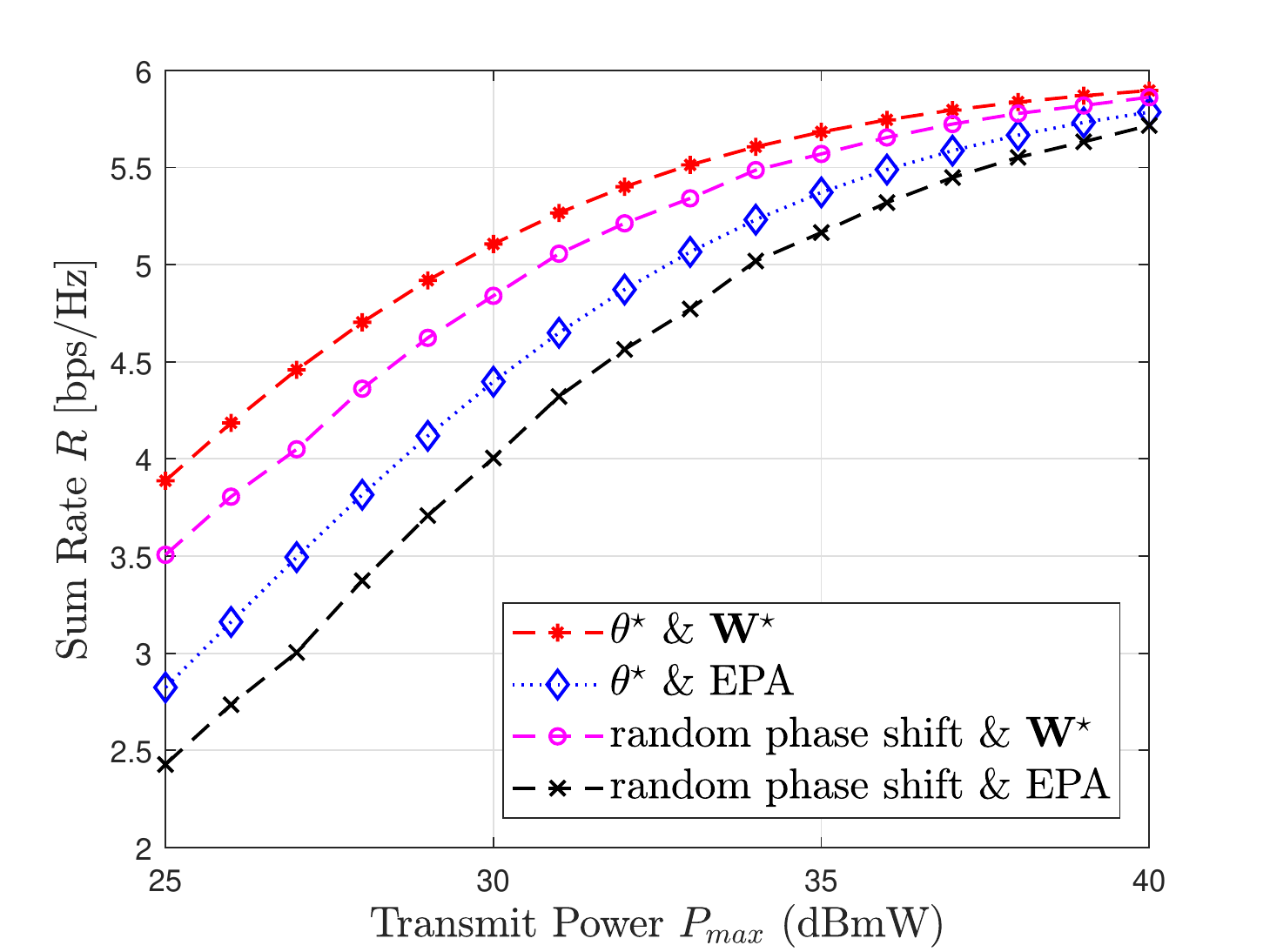}
\caption{The sum-rate versus the BS transmit power for different transmit schemes with $N=80$, $\rho_{\mathbf{T}_B}=\rho_{\mathbf{R}_{U_k}}=0.2$, and $\rho_{\mathbf{T}_I}=\rho_{\mathbf{R}_I}=0.4$. }\label{figure4}
\end{figure}

Fig. \ref{figure4} investigates the sum-rate of different transmit schemes by utilizing  ``random phase shift and EPA" for benchmarking scheme, where the exponential correlation model is adopted. It is readily found that the ``$\theta^\star$ \& $\mathbf{G}^\star$" scheme performs the best among all the transmit schemes and the benchmarking scheme shows the worst performance, which justifies the significance of joint phase shift and precoding design.  Besides, it can be seen from Fig. \ref{figure4} that the ``random phase shift \& $\mathbf{G}^\star$" scheme provides superior performance over the ``$\theta^\star$ \& EGA" scheme. The results reveal that the transmit beamforming (digital beamforming) at the BS is more effective than phase shift design (analog beamforming) at the RIS. Moreover, it can be seen from Fig. \ref{figure4} that  approximate sum-rate can achieve 60\% improvement for lower transmit power ($P_{max}=25 \,\mathrm{dBmW}$), but only achieve 4\% improvement for higher transmit power  ($P_{max}=40 \,\mathrm{dBmW}$) by adopting jointly optimization design. This is due to the fact that the optimization  improvement by (\ref{P_1}) and (\ref{P_bar}) is limited by comparing with transmit power increasing.

\begin{figure}[!htb]
\centering
\includegraphics[width=3.5in]{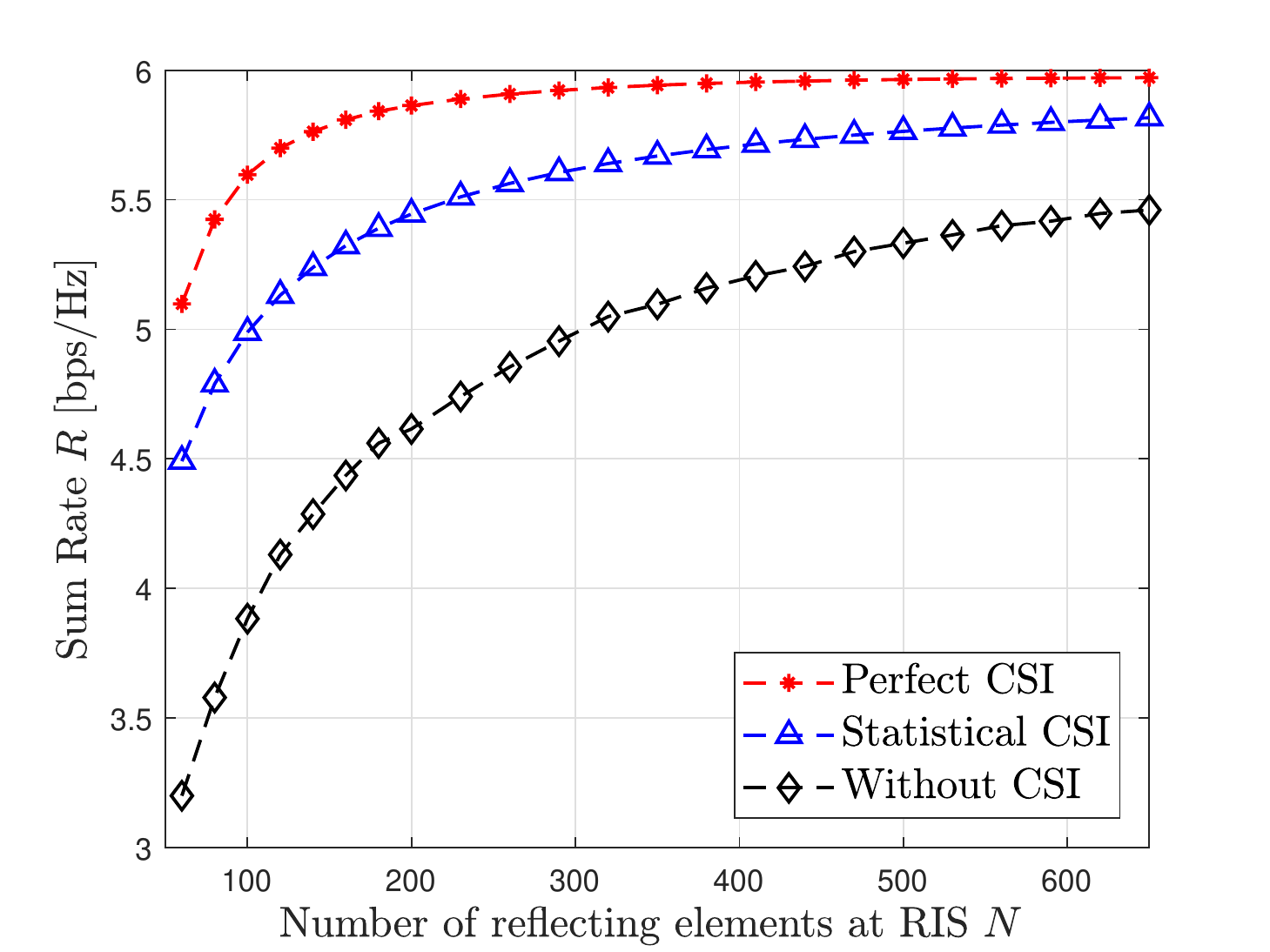}
\caption{The sum-rate versus the number of reflecting elements $N$ with $M=8$ and $P_{max}=30\mathrm{dBmW}$ for various schemes.}\label{figure8}
\end{figure}
In Fig. \ref{figure8}, we compare the sum-rate of RIS-aided multi-user MIMO systems by considering three different CSI conditions, including, perfect CSI, statistical CSI and without CSI. For perfect CSI, we employ the block coordinate descent (BCD) algorithm as \cite{pan2020multicell}, where the precoding matrices and phase shifts are alternately optimized based on instantaneous CSI, and the sum-rate is obtained by averaging over 200 channel realizations. In the case  statistical CSI, the proposed algorithm in Section \ref{new_problem} is adopted. In the case without CSI at the transmitter, the EPA precoding and random phase shift scheme are employed as in \cite{zhang2021large}. It can be observed from Fig. \ref{figure8} that the perfect CSI performs the best among the three schemes at the cost of frequent CSI reportings and high signaling overhead. However, with the increase of the reflecting elements, the gap between the perfect CSI and the statistical CSI decreases due to the channel hardening effect. Although the scheme without CSI is the simplest one for implementation, it shows the worst performance. To strike a balanced trade-off between computational complexity and performance, the proposed algorithm is very appealing particularly for the large scale of RIS-aided systems.

\section{Conclusion}
This paper aimed to enhance the RIS-aided multi-user MIMO system via the joint design of transmit beamforming and phase shifts. However, the spatial correlation has been considered in this paper, which significantly distinguish our analysis and optimal design from the prior ones. To address the above issues, the deterministic equivalent expression of sum-rate was derived by capitalizing on random matrix theory. Thanks to the tractable expression of the approximate sum-rate, some meaningful insights into the impacts of the spatial correlation on RIS-aided multi-user MIMO systems were gained with the aid of majorization theory. Besides, the deterministic equivalent not only enables the accurate evaluation of the sum-rate for a large number of reflecting elements, but also decouples the  transmit beamforming matrix and RIS phase shift vector. This motivated us to formulate a new optimization problem by using the asymptotic sum-rate and exploiting statistical CSI. An efficient algorithm was proposed based on water-filling
and PGA method. Comparing with the designs requiring instantaneous CSI in the literature, our proposed algorithm is more attractive with comparable performance but much lower complexity and communication overhead.
\appendices
\section{Proof of Theorem \ref{theorem1}}\label{app:proof_Theorem1}
In order to derive the asymptotic achievable rate of user $k$, the following lemma regarding the deterministic equivalent of the random matrix chain products is given below.
\begin{lemma}\label{lemma1}
Assume two random matrix with the form $\mathbf{H}_1=\mathbf{R}_{1}^{\frac{1}{2}}\widetilde{\mathbf{H}}_1\mathbf{T}_{1}^{\frac{1}{2}}$
and $\mathbf{H}_2=\mathbf{R}_{2}^{\frac{1}{2}}\widetilde{\mathbf{H}}_2\mathbf{T}_{2}^{\frac{1}{2}}$,
where $\mathbf{T}_{1}\in\mathbb{C}^{M\times M}$, $\mathbf{R}_{1}\in\mathbb{C}^{N\times N}$, $\mathbf{T}_{2}\in\mathbb{C}^{N\times N}$, and $\mathbf{R}_{2}\in\mathbb{C}^{L\times L}$ are Hermitian positive deterministic matrices with bounded spectral norm, i.e., $\mathrm{max}\{\left\| \mathbf{T}_{1}\right\|,\left\| \mathbf{R}_{1} \right\|,\left\| \mathbf{T}_{2}\right\|,\left\|\mathbf{R}_{2} \right\|\}<\infty$ and diagonal elements being one. $\widetilde{\mathbf{H}}_1\in\mathbb{C}^{N\times M}$,  $\widetilde{\mathbf{H}}_2\in\mathbb{C}^{L\times N}$ are two mutually independent random matrices each having i.i.d. entries such that ${\rm{vec}}(\widetilde{\mathbf{H}}_1)\sim{\cal{CN}}({\mathbf{0}_{N\times M }},\sigma_1^2{\mathbf{I}}_{M}\otimes{\mathbf{I}}_{N})$ and ${\rm{vec}}(\widetilde{\mathbf{H}}_2)\sim{\cal{CN}}({\mathbf{0}_{L\times N }},\sigma_2^2{\mathbf{I}}_{N}\otimes{\mathbf{I}}_{L})$. Given a deterministic matrix $\mathbf{W}$ with compatible dimension, then we have the matrix chain products deterministically approaches to
\begin{equation}
\frac{1}{N}\mathbf{H}_2 \mathbf{H}_{1}\mathbf{W}\mathbf{H}_{1}^\mathrm{H}\mathbf{H}_2^\mathrm{H}\xrightarrow[{N_S \to \infty }]{{a.s.}}\frac{\sigma_1^2\sigma_2^2}{N}\mathrm{Tr}\left(\mathbf{T}_2\mathbf{R}_1\right) \mathrm{Tr}\left(\mathbf{T}_1\mathbf{W}\right) \mathbf{R}_{2}.
\end{equation}
\begin{proof}
Lemma \ref{lemma1} can be  obtained
by directly replacing $\mathbf{H}_{1}$ and $\mathbf{H}_{2}$ in \cite[Theorem 1]{feng2019two}, The detailed proof are omitted here to avoid redundancy.
\end{proof}
\end{lemma}

By dividing both the numerator and denominator terms of (\ref{R_k}) by $1/N$, the method of deterministic equivalent is then applied to derive compact expressions for $\mathbf{\Psi}_k/{N}$ and $\mathbf{\Psi}_i/{N}$. Specifically, by using Lemma \ref{lemma1}, it yields a deterministic approximation of $\mathbf{\Psi}_k/{N}$ as
\begin{equation}
\begin{aligned}\label{Psi_k}
\frac{\mathbf{\Psi}_k}{N}=&\frac{1}{N}\mathbf{H}_{IU_{k}}\mathbf{\Theta} \mathbf{H}_{BI}\mathbf{W}_k(\mathbf{H}_{IU_{k}}\mathbf{\Theta} \mathbf{H}_{BI}\mathbf{W}_k)^\mathrm{H}\\
=&\frac{1}{N}\mathbf{R}_{ U_k}^{\frac{1}{2}}\widetilde{\mathbf{H}}_{IU_{k}}\mathbf{T}_{I}^{\frac{1}{2}}\mathbf{\Theta} \mathbf{R}_{I}^{\frac{1}{2}}\widetilde{\mathbf{H}}_{BI}\mathbf{T}_{B}^{\frac{1}{2}}{\mathbf{W}_k}{\mathbf{W}_k^\mathrm{H}}\mathbf{T}_{B}^{\frac{1}{2}}\widetilde{\mathbf{H}}_{BI}^\mathrm{H}\mathbf{R}_{ I}^{\frac{1}{2}}\mathbf{\Theta}^\mathrm{H}\mathbf{T}_{I}^{\frac{1}{2}}\widetilde{\mathbf{H}}_{IU_{k}}^\mathrm{H} \mathbf{R}_{U_k}^{\frac{1}{2}} \\
\xrightarrow[{N \to \infty }]{{a.s.}}&  \beta_{\mathbf{H}_{IU_{k}}}\beta_{\mathbf{H}_{BI}} \frac{\mathrm{Tr}\left(\mathbf{T}_I\mathbf{\Theta}\mathbf{R}_I\mathbf{\Theta}^\mathrm{H}\right)}{N} \mathrm{Tr}\left(\mathbf{T}_B {\mathbf{W}_k}{\mathbf{W}_k^\mathrm{H}}\right) \mathbf{R}_{U_{k}}.
\end{aligned}
\end{equation}
Similarly, the derivation of $\mathbf{\Psi}_i/{N}$ can be obtained by following the same method.
If the number of the reflecting elements is very large, i.e., $N\rightarrow\infty$, the achievable  sum-rate approaches to a deterministic equivalent as (\ref{R_k_asy}) by substituting approximation of $\mathbf{\Psi}_k/{N}$ and $\mathbf{\Psi}_i/{N}$ into (\ref{R_k}). We thus complete the proof.
\bibliographystyle{ieeetran}
\bibliography{RIS}

\end{document}